\documentclass[conference] {IEEEtran}

\usepackage[colorlinks, linkcolor=blue, citecolor=blue]{hyperref}           
\usepackage{color}
\usepackage{graphicx,subfigure,amsmath,amssymb,amsfonts,bm,epsfig,epsf,url,dsfont}
\usepackage{amsthm,mathrsfs}

\newtheorem{thm}{Theorem}[section]
\newtheorem{definition}[thm]{Definition}
\newtheorem{lem}[thm]{Lemma}

\newtheorem{proposition}[thm]{Proposition}
\newtheorem{remark}[thm]{Remark}
\newtheorem{corollary}[thm]{Corollary}

\usepackage{tikz}
\usepackage{bbm}
\usepackage[small,bf]{caption}
\usepackage[top=1in,bottom=1in,left=1in,right=1in]{geometry}
\usepackage{fancybox}
\usepackage{hyperref}
\usepackage{algorithm}
\usepackage{verbatim}
\usepackage{amsmath}
\usepackage{array}
\usepackage{tabularx}
\usepackage{booktabs}

\usepackage{mathtools}

\usepackage{scalerel,stackengine}
\stackMath
\newcommand\reallywidehat[1]{%
\savestack{\tmpbox}{\stretchto{%
  \scaleto{%
    \scalerel*[\widthof{\ensuremath{#1}}]{\kern-.6pt\bigwedge\kern-.6pt}%
    {\rule[-\textheight/2]{1ex}{\textheight}}%WIDTH-LIMITED BIG WEDGE
  }{\textheight}% 
}{0.5ex}}%
\stackon[1pt]{#1}{\tmpbox}%
}

\newcommand*{\rom}[1]{\expandafter\@slowromancap\romannumeral #1@}
\newcommand{\abs}[1]{\left|#1\right|}
\newcommand{\R}{\mathbb{R}}

\DeclareMathOperator*{\argmax}{arg\,max}

\newcommand{\p}[1]{\left(#1\right)}
\newcommand{\pp}[1]{\left[#1\right]}
\newcommand{\ppp}[1]{\left\{#1\right\}}
\newcommand{\norm}[1]{\left\|#1\right\|}

\newcommand{\s}[1]{\mathsf{#1}}
\usepackage{xspace}

\numberwithin{equation}{section}
\newcommand\equalityAS{\mathrel{\stackrel{\makebox[0pt]{\mbox{\normalfont \scriptsize a.s.}}}{=}}}
\newcommand{\N}{N}
\newcommand{\n}{i}
\newcommand{\jj}{j}

\usepackage{setspace}
\IEEEoverridecommandlockouts % Ensures footnotes are visible in conference mode

\allowdisplaybreaks

\begin{document}
\title{A note on the sample complexity of multi-target detection}

\author{
    \IEEEauthorblockN{Amnon Balanov, Shay Kreymer, and Tamir Bendory}
    \thanks{The authors are with the School
of Electrical Engineering, Tel Aviv University, Israel.
    The research is supported by the BSF grant no. 2020159, the NSF-BSF grant no. 2019752, the ISF grant no. 1924/21, and a grant from The Center for AI and Data Science at Tel Aviv University (TAD).}
}

\maketitle
\begin{abstract}
This work studies the sample complexity of the multi-target detection (MTD) problem, which involves recovering a signal from a noisy measurement containing multiple instances of a target signal in unknown locations, each transformed by a random group element.
This problem is primarily motivated by single-particle cryo-electron microscopy (cryo-EM), a groundbreaking technology for determining the structures of biological molecules. 
We establish upper and lower bounds for various MTD models in the high-noise regime as a function of the group, the distribution over the group, and the arrangement of signal occurrences within the measurement.
The lower bounds are established through a reduction to the related multi-reference alignment problem, while the upper bounds are derived from explicit recovery algorithms utilizing autocorrelation analysis.
These findings provide fundamental insights into estimation limits in noisy environments and lay the groundwork for extending this analysis to more complex applications, such as cryo-EM. 

\end{abstract}

\section{Introduction}

 We study the multi-target detection (MTD) problem, where multiple signals are embedded at unknown positions within a long, noisy observation $y$~\cite{bendory2019multi}.
Let $x\in V$ be the target signal to be estimated, where  $V$ is a finite-dimensional vector space. To present the problem, we consider $V=\R^L$; however, we later extend the discussion to include band-limited images.
Let $g_i\in G$ be a group element drawn i.i.d. from a distribution $\rho$ over \(G\), and let $x_i = g_i\cdot x \in \mathbb{R}^L$.
An MTD observation $y: \ppp{0,1,\dots, LM-1} \to \mathbb{R}$ is given by 
\begin{align}
    y = \sum_{\n=0}^{\N-1} {s_{\n} * x_{\n}} + \epsilon, \label{eqn:MTDmodelGeneral}
\end{align}  
where $*$ denotes linear convolution, $\N$ is the number of signal occurrences, and \(\epsilon \sim \mathcal{N}(0, \sigma^2 I_{LM \times LM})\).
The unknown locations of the signal occurrences are represented by location indicators \(\{s_{\n}\}_{\n=0}^{\N-1}\). Each \(s_{\n}\) is a one-hot vector, where \(s_{\n}[\jj] = 1\) specifies that the first entry of \(x_\n\) is placed at entry \(\jj\) of the observation $y$. The locations \(\{s_\n\}_{\n=0}^{\N-1}\) are assumed to be deterministic, but unknown.
Since the statistics of \(y\) is the same for $x$ and $g\cdot x$ for any $g\in G$, we aim to recover the $G$-orbit of $x$ $\{g\cdot x\,|\, g\in G\}$.

\textbf{Motivation: Cryo-EM.} 
A key application of the MTD model is single particle cryo-electron microscopy (cryo-EM) \cite{nogales2016development, bendory2020single,singer2020computational}. In cryo-EM, 3-D biological samples, such as macromolecules or viruses, are rapidly frozen in a thin layer of vitreous ice (Figure \ref{fig:1}(a)). The specimen is then imaged using an electron microscope that captures 2D tomographic images, known as \textit{micrographs}. Each micrograph contains multiple tomographic projections of the molecules, where the 3-D orientation and 2-D position of each molecule are unknown. Therefore, cryo-EM models can be described similarly to a 2-D version of~\eqref{eqn:MTDmodelGeneral}, where each signal occurrence is of the form $x_\n = \Pi(g_\n \cdot x)$, where $g_\n$ is a 3-D rotation and $\Pi$ is a tomographic projection. (Note that the tomographic projection is not a group action, and thus~\eqref{eqn:MTDmodelGeneral} does not fully capture the cryo-EM model.)  

\begin{figure*}[!t]
    \centering
    \includegraphics[width=0.95 \linewidth]{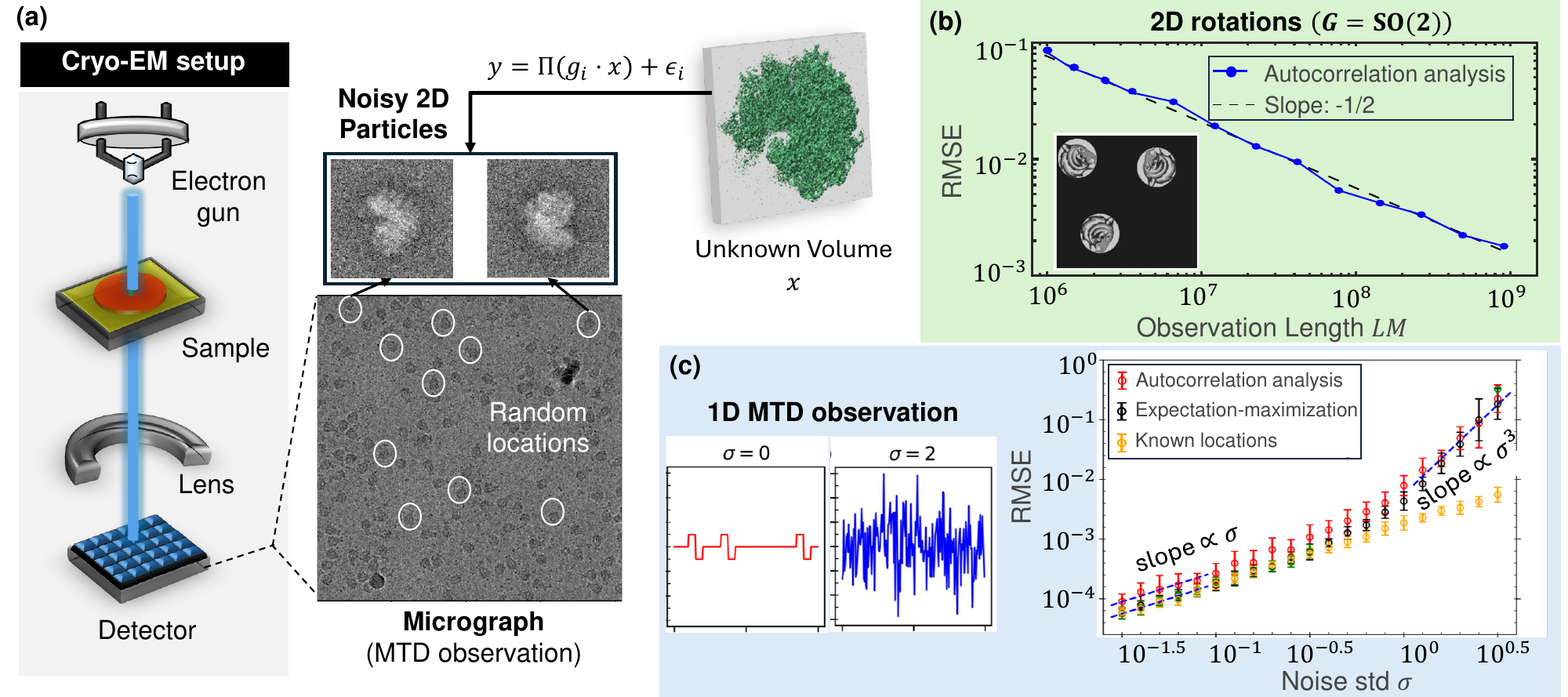}
    \caption{\textbf{(a)} Single-particle electron microscopy reconstructs 3D structures from 2D projections~\cite{bendory2020single}. Particles in vitrified ice form a \textit{micrograph}, modeled using the MTD framework~\eqref{eqn:MTDmodelGeneral}, where \( x_\n = \Pi(g_\n \cdot x) \), with \( g_\n \) a 3D rotation and \( \Pi \) as a tomographic projection. In high noise levels, direct particle detection is infeasible, but 3D reconstruction might be feasible by directly processing the micrographs~\cite{bendory2023toward}.  \textbf{(b)} The estimation error (RMSE) versus observation length for MTD with 2D rotations based on autocorrelation analysis (taken from~\cite{kreymer2022two}). As seen empirically, the image can be estimated using the third-order autocorrelation. \textbf{(c)} An MTD observation of 1D signals in which identical signals \(x_\n = x\) are located at unknown positions in the observation $y$. In the high noise regime, their locations cannot be estimated, but the signal \(x\) can be estimated accurately. The estimation error (RMSE) scales as $\sigma^3$ in high noise levels, consistent across autocorrelation analysis and expectation-maximization (taken from~\cite{lan2020multi}).}
    \label{fig:1}
\end{figure*}

Traditional methods for reconstructing the 3D structure from micrographs typically involve a two-step process: detecting and extracting the projection images from the micrographs, followed by 3-D reconstruction from those images. 
However, reliably detecting individual particles is challenging in high-noise environments. This is especially important for small molecular structures that induce high noise levels. 
Therefore, it is typically claimed that recovering small molecular structures using cryo-EM is impossible~\cite{henderson1995potential}. 
In this context, the MTD model was introduced as a computational framework to circumvent particle picking by directly reconstructing the 3D structure from the micrograph. Consequently, the inability to perform particle picking does not imply the impossibility of recovering the molecular structure, thereby enabling the recovery of small molecular structures~\cite{bendory2023toward,kreymer2023stochastic}. 
This technique also mitigates the ubiquitous bias issues in particle picking~\cite{henderson2013avoiding,balanov2024einstein}. 

\textbf{Contribution.} 
The main goal of this study is to derive upper and lower bounds on the sample complexity of the MTD problem in the high noise regime, focusing on three specific instances.
Section~\ref{sec:preliminaries} provides some basic definitions, and Section \ref{sec:autoCorrealtionAnalysis} introduces autocorrelation analysis, which is required for deriving the upper bounds of sample complexity. Section \ref{sec:mra} presents the multi-reference alignment (MRA) model and its tight connection to the MTD model, which is the cornerstone of the sample complexity lower bounds. Section~\ref{sec:results} integrates these findings, focusing on three specific MTD models, and Section~\ref{sec:outlook} delineates how these techniques could be extended to derive the sample complexity of the cryo-EM technology, the ultimate goal of this research.
The proofs of the results are relegated to the appendix. 
% Due to space constraints, the proofs are omitted from this manuscript and are available in [].

\section{Preliminaries} \label{sec:preliminaries}

\subsection{Assumptions on the MTD model}
We consider the asymptotic case of \(\sigma, \N, M \to \infty\), such that the density \(\gamma = \N / M<1\) remains fixed. 
We assume that \(\gamma\) and the noise variance \(\sigma^2\) are known. %\shay{Maybe to add why? Or we don't want to dig in to it?}

Throughout, we assume that the signal occurrences do not overlap. This, in turn, means that their starting indices are separated by at least $L$ entries so that their supports do not intersect, and the corresponding signal instances do not interfere with each other within the observation \( y \).
In some cases, we consider a stronger assumption that the minimal separation is doubled and refer to this model as the well-separated case. 
In this case,  the starting positions of any two occurrences must be separated by at least \( 2L \) positions, ensuring that their endpoints are separated by at least \( L \) signal-free entries (although still noisy) in the data. 

\begin{definition} [Separation] \label{def:wellSeperatedModel}
Consider the following separation condition:
\begin{align} \label{eq:sep}
  \nonumber \text{If} \ s_{\n_1}[\jj_1] = 1 \ \text{and} \ s_{\n_2}[\jj_2] = 1 \ \text{for} \ \n_1 \neq \n_2, \\
   \text{then} \ \abs{\jj_1 - \jj_2} \geq L_{\text{sep}}.
\end{align}
If the MTD model satisfies~\eqref{eq:sep} with $L_{\text{sep}}\geq L$, we say it is the non-overlapping; if $L_{\text{sep}}\geq 2L$, we say it is well-separated. 
\end{definition}

\begin{remark}
    The results of this paper regarding the well-separated model are valid for $L_{\text{sep}} \geq 2L - 1$. Nevertheless, for the sake of simplifying the expressions throughout the text, we assume a separation of $2L$.
\end{remark}

\subsection{Sample complexity definition}
The objective of this work is to determine the sample complexity, i.e., the minimum number of measurements required to achieve a desired mean-squared error (MSE) between the true signal $x$ and its estimator $\hat{x}$=$\hat{x}(y)$. Importantly, as the estimation of $x$ is considered up to a group action on $x$, the MSE is defined as~\cite{abbe2018multireference}: 
\begin{equation}
    \text{MSE}\p{\hat{x}, x} = \frac{1}{\norm{x}_F^2} \mathbb{E} \pp{\min_{g \in G} \norm{g \cdot \hat{x} - x}_F^2}. \label{eqn:mseDef}
\end{equation}

Then, the sample complexity is defined as follows:
\begin{definition} [Sample complexity] \label{def:sampleComplexity}
    Suppose that all the parameters of the MTD model \eqref{eqn:MTDmodelGeneral} are fixed except for $\N,M$, and $\sigma$. We define the smallest MSE attainable by any estimator as
    \begin{equation}
        \s{MSE}^*_{\s{MTD}}(\sigma^2, \N) := \inf_{\hat{x}} \mathbb{E} \pp{ \s{MSE} \p{\hat{x}(y), x}}, \label{eqn:infMtdMSE}
    \end{equation}
    and the sample complexity as,
    \begin{equation}
        \N^*_{\s{MTD}} \p{\sigma^2, \epsilon} := \min \left\{ \N : \s{MSE}^*_{\s{MTD}}(\sigma^2, \N) \leq \epsilon \right\}. \label{eqn:sampleComplexityMTD}
    \end{equation}
\end{definition}
In the above definition, $\epsilon$ represents a desired error, which in this paper approaches zero.

\section{Autocorrelation analysis} \label{sec:autoCorrealtionAnalysis}
In this section, we present the autocorrelation analysis method and provide several results that will be essential in establishing the upper bounds for sample complexity.
To address the difficulty of locating signals at high noise levels~\cite{aguerrebere2016fundamental,dadon2024detection},
previous works suggested the use of features that are invariant to unknown locations and possibly to the unknown group actions~\cite{bendory2019multi,bendory2023toward,kreymer2022two,lan2020multi,bendory2023multi}.

\begin{definition} [Empirical autocorrelations] \label{def:autoCorrelationMomentsEmpirical}
    Let $z \in \mathbb{R}^{LM}$ be an observation following the MTD model in \eqref{eqn:MTDmodelGeneral}. The empirical $d$-order autocorrelation of $z$ is defined by: 
\begin{align}
    \nonumber {a}^{(d)}_{z} & [\ell_1, \ell_2, ..., \ell_{d-1}] 
    \\ & = \frac{1}{LM} \sum_{\jj=0}^{LM-1} {\tilde{z}[\jj]\tilde{z}[\jj+\ell_1] ... \tilde{z}[\jj+\ell_{d-1}]}, \label{eqn:autoCorrealtionMoments}
\end{align}
for $0 \leq \ell_1, \ell_2, ..., \ell_{d-1}  \leq L-1$, where $\tilde{z}$ is the padding with zeros of $z$ to a length of $\p{M+1}L$.
\end{definition}

\begin{definition} [Autocorrelation ensemble mean] \label{def:autoCorrelationMomentsEnsembele}
    Let \( X_G = g \cdot x\in  \mathbb{R}^L \), where \( g \in G \) is a random group element drawn from the distribution \( g \sim \rho \).
     Define
    \begin{align}
        Y = \tilde{X}_G + \epsilon \in \mathbb{R}^{3L}, \label{eqn:ensembeleRV}
    \end{align}
    where $\tilde{X}_G = [\bold{0}_L, X_G, \bold{0}_L]$ is the padding of $X_G$ with zeros from left and right by $L$ zeros each, and $\epsilon \sim \mathcal{N} \p{0, \sigma^2 I_{3L \times 3L}}$.
    Then, the $d$-order autocorrelation ensemble mean of $Y$ is defined by,
\begin{align}
    \nonumber \bar{a}^{(d)}_{Y,\rho} & [\ell_1, \ell_2, ..., \ell_{d-1}] 
    \\ & = \frac{1}{2L} \sum_{\jj=0}^{2L-1} \mathbb{E}_{G \sim \rho, \epsilon} \ppp{Y[\jj]Y [\jj+\ell_1] ... Y[\jj+\ell_{d-1}]}, \label{eqn:autoCorrealtionMomentsEnsembele}
\end{align}
for $0 \leq \ell_1, \ell_2, ..., \ell_{d-1}  \leq L-1$. 
\end{definition}

The following proposition
shows, for the well-separated case, that the empirical autocorrelation converges almost surely to the autocorrelation ensemble mean. This result is an extension of~\cite[Appendix A]{bendory2019multi}. 

\begin{proposition}
\label{thm:prop0}
Let $\s{MTD}_G$ be a well-separated MTD model defined in \eqref{eqn:MTDmodelGeneral} with a compact group $G$ acting on $\mathbb{R}^L$, with group elements drawn according to a distribution $\rho$. Let ${a}^{(d)}_{z}$ be the empirical $d$-order autocorrelation as defined in \eqref{eqn:autoCorrealtionMoments}, and $\bar{a}^{(d)}_{Y,\rho}$ be the $d$-order autocorrelation ensemble mean as defined in \eqref{eqn:autoCorrealtionMomentsEnsembele}. Then, for every $0 \leq \ell_1, \ell_2, ..., \ell_{d-1}  \leq L-1$,
\begin{align}
    \nonumber \lim_{\N,M \to \infty} & {a}^{(d)}_{z} [\ell_1, \ell_2, ..., \ell_{d-1}]
    \\ \nonumber \equalityAS & \ 2 \gamma \bar{a}^{(d)}_{Y,\rho} [\ell_1, \ell_2, ..., \ell_{d-1}]
    \\ & + \p{1-2 \gamma} \chi [\ell_1, \ell_2, ..., \ell_{d-1}],
\end{align}
where 
\begin{align*}
    \chi[\ell_1, \ell_2, ..., \ell_{d-1}] = \mathbb{E} \pp{\epsilon \pp{0} \epsilon \pp{\ell_1} , ... \epsilon \pp{\ell_{d-1}}}.
\end{align*}
\end{proposition} 

Next, we show that if the autocorrelations up to order \( \bar{d} \) uniquely define an orbit of the signal, then the sample complexity is bounded above by \( \omega \big(\sigma^{2\bar{d}}\big) \).

\begin{proposition}
\label{thm:prop1}
Assume the same conditions as in Proposition \ref{thm:prop0}. In addition, assume the parameter space $\Theta$ of the unknown signals $x \in \Theta$ is compact. Then, if the autocorrelation ensemble mean up to order $\bar{d}$, as defined in \eqref{eqn:autoCorrealtionMomentsEnsembele}, uniquely determines the orbit of the signal, then the sample complexity of recovering the orbit is upper bounded by $\omega \p{\sigma^{2 \bar{d}}}$.
\end{proposition}

\section{The multi-reference alignment model} \label{sec:mra}
\subsection{The multi-reference alignment model}
To establish the lower bounds for the MTD models, we first introduce
the multi-reference alignment (MRA) model~\cite{bandeira2014multireference}. 
Formally, the MRA model is defined as,
\begin{equation}
    z_\n = g_\n \cdot x + \epsilon_\n, \quad g_\n \in G, \label{eqn:MRA_model}
\end{equation}
where $G$ is a known compact group, the group elements are drawn from some distribution $\rho$ over $G$,  and $\epsilon_\n \sim \mathcal{N}(0, \sigma^2 I)$. The objective is to estimate \(x \in \mathbb{R}^L\) from \(\N\) observations \(\{z_\n\}_{\n=0}^{\N-1}\). From the perspective of this paper, the MRA problem can be interpreted as a simplified model of the MTD problem, when the signal locations are known. The sample complexity of the MRA problem can be defined analogously to the sample complexity of Definition~\ref{def:sampleComplexity}, where the number of observations $\N$ replaces the number of signal occurrences, and the observed data is now \(\{z_\n\}_{\n=0}^{\N-1}\). We denote the sample complexity of MRA by $\N^*_{\s{MRA}_G} \p{\sigma^2, \epsilon}.$

\subsection{Mapping between the MRA model and the associated MTD model.}
Comparing the MRA model \eqref{eqn:MRA_model} with the MTD model \eqref{eqn:MTDmodelGeneral} shows that the MTD model has an additional degree of uncertainty due to the unknown locations. The following proposition shows that the sample complexity of the corresponding MRA problem lower bounds the sample complexity of the MTD model.

\begin{proposition}
\label{thm:prop2}
Consider an MRA model \eqref{eqn:MRA_model}, with a distribution \(\rho\) defined on a group $G$. Consider a corresponding MTD model \eqref{eqn:MTDmodelGeneral} with the same group \( G \) and distribution \(\rho\), with \( \N \) signal occurrences located at unknown positions \(\{s_\n\}_{\n=0}^{\N-1}\), assuming the non-overlapping case (Definition \ref{def:wellSeperatedModel}). Then, the sample complexity of the MTD model \eqref{eqn:MTDmodelGeneral} is lower bounded by the sample complexity of the associated MRA model \eqref{eqn:MRA_model}, i.e.,
\begin{align}
     N^*_{\s{MTD}_G} \p{\sigma^2, \epsilon} \geq  N^*_{\s{MRA}_G} \p{\sigma^2, \epsilon} \label{eqn:prop2}.
\end{align}
\end{proposition}

\section{Sample complexity bounds} \label{sec:results}
We now establish lower and upper bounds on the sample complexity for three specific MTD instances. % We briefly outline the proof strategy to clarify the connection to earlier sections. 
The lower bounds are derived by reducing the problem to the corresponding MRA models and leveraging existing results. The upper bounds are obtained through explicit recovery techniques based on autocorrelation analysis.

\subsection{One-dimensional MTD with circular translations}
The first MTD model considers $x\in\R^L$, $G$ is the group of circular translations $G=\mathbb{Z}_L$,  and $\rho$ is a uniform distribution. We demonstrate that both the upper and lower bounds for the sample complexity are $\omega \p{\sigma^6}$. 

\begin{proposition}
 \label{thm:thm1DsignalTranslations}
Consider the MTD model with a uniform distribution $\rho$ of the group $G = \mathbb{Z}_L$ elements, acting on $x\in\R^L$. Assume that the signal $x$ has a non-vanishing DFT. Then, as $\sigma, \N \to \infty$, we have,
\begin{enumerate}
    \item The sample complexity of the well-separated case is bounded from above by  $\omega \p{{\sigma^6}}$. 
    \item The sample complexity of the non-overlapping case is bounded from below by $\omega \p{{\sigma^6}}$.
\end{enumerate}
   \end{proposition}

We remark that if $\rho$ is non-uniform, then the upper bound remains $\omega \p{{\sigma^6}}$ and the lower bound is $\omega \p{{\sigma^4}}$~\cite{abbe2018multireference}.

\subsection{MTD with two-dimensional rotations}
Next, we consider the MTD problem of estimating a two-dimensional \emph{band-limited} target image $x$ acted upon by SO(2) rotations, drawn from a uniform distribution.  
By bandlimited images, we mean it can be represented by finitely many Fourier-Bessel coefficients or alternative steerable basis coefficients~\cite{marshall2023fast}.
Importantly, although the results from the previous sections pertain to signals in $\R^L$, they can be naturally extended to this case with suitable modifications.

\begin{proposition}
\label{thm:thm2Drotations}
Consider the MTD model of estimating non-overlapping images, acted upon by SO(2) rotations drawn from a uniform distribution. 
Assume that the image $x$ has a finite expansion in a steerable basis (e.g., Fourier-Bessel) and all the coefficients are non-zero. Then, as $\sigma, \N \to \infty$, the sample complexity is lower bounded by  $\omega \p{{\sigma^6}}$.
\end{proposition}

Previous empirical findings suggest that the third autocorrelation moment is sufficient to reconstruct the two-dimensional image~\cite{ma2019heterogeneous, kreymer2022two}; see Figure~\ref{fig:1}(b). Another study~\cite{kreymer2022approximate} showed that an approximate expectation-maximization method achieves comparable results. Based on this, we conjecture an upper bound of $\omega \p{\sigma^6}$. 

\subsection{MTD with one-dimensional single signal}
We now consider the one-dimensional MTD model with no group action, $G=I$. In this case, the sample complexity is upper bounded by  $\omega \p{\sigma^6}$.

\begin{proposition}
\label{thm:thm1Dsignal}
Consider the MTD model for $x\in\R^L$ and $G=I$.  Assume that the signal $x$ has non-vanishing entries.   
Then, as $\sigma, \N \to \infty$, the sample complexity is upper bounded by $\omega \p{{\sigma^6}}$.
\end{proposition}

The lower bound for this model has not yet been proven. Based on previous empirical results~\cite{lan2020multi} (see Figure \ref{fig:1}(c)), we conjecture a lower bound of $\omega(\sigma^6)$. %If confirmed, this would establish a tight bound for the model.

\section{Outlook} \label{sec:outlook}
This work is a first step toward analyzing the sample complexity of the MTD model, where the ultimate goal is to understand the sample complexity of the cryo-EM. 

\textbf{Cryo-EM and related models.} The lower bounds in this work are derived by reducing the MTD model to a simpler MRA model. Consequently, any findings on the sample complexity of the MRA problem directly translate into lower bounds for the MTD model. In particular, in~\cite{bandeira2023estimation}, the authors derived the sample complexity of many different MRA models and cryo-EM for generic signals in some specific parameter regimes. In particular, this paper implies a lower bound of $\omega\p{\sigma^6}$ on the sample complexity of the cryo-EM problem, at least in some parameter regime. 
Similarly, for signals in $\R^L$ acted upon by elements of the dihedral group, 
the results of~\cite{bendory2022dihedral,edidin2024generic}  imply a lower bound of $\omega\p{\sigma^6}$ if the distribution over the group is uniform, and $\omega\p{\sigma^4}$ if the distribution is non-uniform.

\textbf{Priors.} The analysis of the lower bounds on the sample complexity can be extended to include a prior on the signal. For instance, in~\cite{romanov2021multi}, the sample complexity of high-dimensional MRA under a Gaussian prior has been derived. Similarly, in~\cite{bendory2024transversality}, the sample complexity of signals that reside in a low-dimensional semi-algebraic set was studied.
This generalization underscores the robustness of the derived bounds, demonstrating their applicability even when additional prior information about $x$ is included in the model.

\textbf{Upper bounds via autocorrelation analysis.} The upper bounds are based on explicit algorithms for signal recovery from autocorrelations. To date, very few provable MTD algorithms have been designed. Designing such algorithms---besides their importance for validation, reproducibility,  and consistency---will immediately imply sample complexity upper bounds.

\bibliographystyle{plain}
%\bibliography{refs} 

%\begin{appendices}
\appendix
\section{Proofs}
\subsection{Proof of Proposition \ref{thm:prop0}}
Let \( z \in \mathbb{R}^{LM} \) represent observations sampled from the well-separated MTD model~\eqref{eqn:MTDmodelGeneral}. 
The indices \( \ppp{r_0, r_1, r_2, \ldots, r_{\N-1}} \) denote the starting indices of the signals \( \{x_\n\}_{\n=0}^{\N-1} \) in $z$, where \( s_\n[r_\n] = 1 \) for every $i \in \pp{N}$ and zeros otherwise. 
According to the well-separated assumption (Definition~\ref{def:wellSeperatedModel}), the distance between consecutive indices satisfies \( \abs{r_m - r_{m-1}} \geq 2L \).

Now, define two sets of random vectors $\ppp{A_\n}_{\n=0}^{\N-1}$, and $\ppp{B_\n}_{\n=0}^{\N-1}$, by,
\begin{align}
    A_{\n} = \p{z\pp{r_\n}, z\pp{r_\n+1}, \dots, z\pp{r_{\n}+2L-1}},
    \label{eqn:seqAi}
\end{align}
and,
\begin{align}
    B_{\n} = \p{z\pp{r_\n-L}, \dots, z\pp{r_\n}, \dots, z\pp{r_{\n}+L-1}}.
    \label{eqn:seqBi}
\end{align}
Note that $A_{\n}, B_{\n} \in \mathbb{R}^{2L}$.
These definitions correspond to slices of the observation \( z \):
\begin{itemize}
    \item $A_\n$ is a segment of length \( 2L \), starting from the beginning of the \( \n \)-th signal instance $x_\n$.
    \item $B_\n$ is also a segment of length \( 2L \), but it begins \( L \) entries before the starting position of the \( \n \)-th signal instance $x_\n$.
\end{itemize}
Then, we have the following characterization of the sequences of $A_\n, B_\n$ for $\n \in \pp{\N}$, proven at this section's end.
\begin{lem} \label{lemma:C00} 
Assume the well-separated MTD model (Definition \ref{def:wellSeperatedModel}). Then, the sequence  
$\ppp{A_\n}_{\n=0}^{\N-1}$ is i.i.d., and follows the distribution of,
\begin{align}
    A = [g \cdot x, \bold{0}_L] + \epsilon \in \mathbb{R}^{2L}, \label{eqn:ArV}
\end{align}
where $\bold{0}_L$ denotes a vector of zeros of length $L$, $g \in G$ is drawn according to the distribution law $\rho$, and $\epsilon \sim \mathcal{N} \p{0, \sigma^2I_{2L \times 2L}}$.
Similarly, the sequence  
$\ppp{B_\n}_{\n=0}^{\N-1}$ is i.i.d., and follows the distribution of,
\begin{align}
    B = [ \bold{0}_L, g \cdot x] + \epsilon \in \mathbb{R}^{2L}. \label{eqn:BrV}
\end{align}
 
\end{lem}

By definition, the distribution of $A$ \eqref{eqn:ArV} is equal to the distribution of the ensemble random vector $Y$ \eqref{eqn:ensembeleRV} from the index $L$ to the index $3L-1$:
\begin{align}
    A = \p{Y\pp{L}, \dots, Y\pp{3L-1}}. \label{eqn:app_A5}
\end{align}
Similarly, the distribution of $B$ \eqref{eqn:BrV} is equal to the distribution of $Y$ \eqref{eqn:ensembeleRV} from the index $0$ to the index $2L-1$:
\begin{align}
    B = \p{Y\pp{0}, \dots, Y\pp{2L-1}}. \label{eqn:app_A6}
\end{align}

To simplify notation, let us define
\begin{align}
    F \pp{\jj} = z\pp{\jj}z\pp{\jj+\ell_1} ...z\pp{\jj+\ell_{d-1}}, \label{eqn: app_C0}
\end{align}
for $0 \leq \ell_1, \ell_2, ..., \ell_{d-1}  \leq L-1$. With this definition, the empirical autocorrelation of order $d$ \eqref{eqn:autoCorrealtionMoments} can be expressed as,
\begin{align}
    {a}^{(d)}_{z} & [\ell_1, \ell_2, ..., \ell_{d-1}]  = \frac{1}{LM} \sum_{\jj=0}^{LM-1} F\pp{j}. \label{eqn:autoCorrelationInFrep}
\end{align}

It is important to note that $F\pp{\jj}$ is generally not an i.i.d. sequence for different $\jj$'s. As a result, the strong law of large numbers (SLLN) cannot be applied directly on \eqref{eqn:autoCorrelationInFrep}. To address this, we decompose the sum in \eqref{eqn:autoCorrelationInFrep} into several distinct infinite sums, each containing i.i.d.\ terms. The SLLN is then applied to each one of these infinite sums separately.

By the definition of $F\pp{\jj}$, and the starting positions $\ppp{r_\n}_{\n=0}^{\N-1}$, for every $0 \leq k < L-1$, the term $F\pp{r_\n + k}$ depends solely on the vector $A_\n$, as defined in \eqref{eqn:seqAi}. Similarly, for $-L \leq k < 0$, the term $F\pp{r_\n + k}$ depends solely on the vector $B_\n$, as defined in \eqref{eqn:seqBi}.
Based on this and Lemma \ref{lemma:C00}, we make the following observation, which is proven later on.

\begin{lem} \label{lemma:C0} 
Assume the well-separated MTD model, and recall the definition of $F\pp{\jj}$ in \eqref{eqn: app_C0}. Fix $-L \leq k < L$. Then, the sequence $\ppp{F\pp{r_\n + k}}_{\n=0}^{\N-1}$ is i.i.d, where $\ppp{r_\n}_{\n=0}^{\N-1}$ are the beginning of the locations of the signals, $s_\n[r_\n] = 1$.
\end{lem}
As a consequence of Lemma~\ref{lemma:C0}, we have the following result:

\begin{corollary} \label{lemma:C1} 
Assume the well-separated MTD model, and recall the definition of $F\pp{\jj}$ in \eqref{eqn: app_C0}. 
Then,  
    \begin{align}
        \frac{1}{LM} & \sum_{\jj=0}^{LM-1} F\pp{\jj} \xrightarrow{\text{a.s.}} 
        \\  \frac{\gamma}{L} & \sum_{k=-L}^{L-1}\mathbb{E} \pp{F\pp{r_0 + k}} 
        \\ & + \p{1 - 2\gamma} \mathbb{E} \pp{\epsilon \pp{0} \epsilon  \pp{\ell_1} ...\epsilon \pp{\ell_{d-1}}}, \label{eqn:app_A11}
    \end{align}
as $\N, M \to \infty$.    
\end{corollary} 

The following lemma characterizes the terms $\mathbb{E}\pp{F\pp{r_0 + k}}$ in \eqref{eqn:app_A11} for every $-L \leq k < L$,  with respect to the definitions of the autocorrelation ensemble mean (Definition \ref{def:autoCorrelationMomentsEnsembele}).
\begin{lem} \label{lemma:C2} 
Assume the well-separated MTD model. Recall the definition of the autocorrelation ensemble mean $\bar{a}^{(d)}_{Y,\rho}$ (Definition \ref{def:autoCorrelationMomentsEnsembele}), Then, 
\begin{align}
    \frac{1}{2L} \sum_{k=-L}^{L-1}\mathbb{E} \pp{F\pp{r_0 + k}} = \bar{a}^{(d)}_{Y,\rho},
\end{align}
\end{lem}
Combining \eqref{eqn:autoCorrelationInFrep}, Lemma \ref{lemma:C1} and Lemma \ref{lemma:C2} completes the proof. It is left to prove the Lemmas.

\begin{proof}[Proof of Lemma~\ref{lemma:C00}]
By the construction of the MTD model \eqref{eqn:MTDmodelGeneral}, the stochastic process arises from two distinct components independent of each other: the group actions $\ppp{g_\n}_{\n=0}^{\N-1} \in G$, drawn from a distribution $\rho$ and acting on $x$, and the additive Gaussian noise, $\epsilon \sim \mathcal{N}\p{0, \sigma^2I}$. In contrast, the signal $x$ and the locations $\ppp{s_\n}_{\n=0}^{\N-1}$ are deterministic but unknown. Thus, it suffices to show that the sequence $\ppp{A_\n}_{\n=0}^{\N-1}$ is i.i.d.\ with respect to these two stochastic processes.

Indeed, by the well-separated assumption, for every $\n_1 \neq \n_2$, $A_{{\n_1}}$ and $A_{{\n_2}}$ do not overlap. Consequently, they are influenced by distinct realizations of the additive Gaussian noise and group actions, $g_{\n_1}$ and $g_{\n_2}$. By the definition of the MTD model, these components are i.i.d., resulting in $A_{{\n_1}}$ and $A_{{\n_2}}$ being independent. 

Similarly, for every $\n_1 \neq \n_2$, $B_{{\n_1}}$ and $B_{{\n_2}}$ do not overlap. Therefore, by the same reasoning as for $A_{\n_1}$ and $A_{\n_2}$, $B_{{\n_1}}$ and $B_{{\n_2}}$ are independent.

Moreover, under the well-separated MTD model, each signal instance $x_\n$ is padded with $L$ zeros on both sides. Thus, by definition, the sequences $\ppp{A_\n}_{\n=0}^{\N-1}$ and $\ppp{B_\n}_{\n=0}^{\N-1}$, as defined in \eqref{eqn:seqAi} and \eqref{eqn:seqBi}, follow the distributions specified in \eqref{eqn:ArV} and \eqref{eqn:BrV}, respectively.

\end{proof}

\begin{proof}[Proof of Lemma~\ref{lemma:C0}]
By the definition of \( F\pp{\jj} \), for every \( 0 \leq k < L \), the term \( F\pp{r_\n + k} \) depends solely on the vector \( A_\n \), as defined in \eqref{eqn:seqAi}:
\begin{align}
    F\pp{r_\n + k} = A_\n \pp{k}A_\n \pp{k+\ell_1} ... A_\n \pp{k + \ell_{d-1}}. \label{eqn:app_A13}
\end{align}

Furthermore, for any \( \n_1 \neq \n_2 \), the dependence of \( F\pp{r_{\n_1} + k} \) on the entries of \( A_{\n_1} \) is identical to the dependence of \( F\pp{r_{\n_2} + k} \) on the entries of \( A_{\n_2} \), i.e.,
\begin{align}
    \nonumber F\pp{r_{\n_1} + k} = A_{\n_1} \pp{k}A_{\n_1} \pp{k+\ell_1} ... A_{\n_1} \pp{k + \ell_{d-1}}. 
    \\ F\pp{r_{\n_2} + k} = A_{\n_2} \pp{k}A_{\n2} \pp{k+\ell_1} ... A_{\n_2} \pp{k + \ell_{d-1}}. 
    \label{eqn:app_A14}
\end{align}
Since \( \{A_\n\}_{\n=0}^{\N-1} \) is an i.i.d sequence (Lemma \ref{lemma:C00}), it follows that \( \{F\pp{r_\n + k}\}_{\n=0}^{\N-1} \) is also an i.i.d. sequence.

Similarly, for every \( -L \leq k < 0 \), the term \( F\pp{r_\n + k} \) depends solely on the vector \( B_\n \), as defined in \eqref{eqn:seqBi}:
\begin{align}
    F\pp{r_\n + k} = B_\n \pp{k}B_\n \pp{k+\ell_1} ... B_\n \pp{k + \ell_{d-1}}. \label{eqn:app_A15}
\end{align}
Additionally, for any \( \n_1 \neq \n_2 \), the dependence of \( F\pp{r_{\n_1} + k} \) on the entries of \( B_{\n_1} \) is identical to the dependence of \( F\pp{r_{\n_2} + k} \) on the entries of \( B_{\n_2} \). Since \( \{B_\n\}_{\n=0}^{\N-1} \) is an i.i.d. sequence, it follows that \( \{F\pp{r_\n + k}\}_{\n=0}^{\N-1} \) is also an i.i.d. sequence. Thus, for every $-L \leq k < L-1$,  \( \{F\pp{r_\n + k}\}_{\n=0}^{\N-1} \) is an i.i.d. sequence, completing the proof.
\end{proof}

\begin{proof}[Proof of Corollary~\ref{lemma:C1}]
We define the set,
\begin{align}
    \mathcal{C}_M^{(\n)} = \ppp{m \in \pp{LM}: \abs{m - r_\n} < L} \subset \pp{LM},
\end{align} 
i.e., the indices that are located at most $L$ indices relative to the beginning of the signal $x_i$. In addition, we define
\begin{align}
    \mathcal{C}_M = \bigcup_{\n=0}^{\N-1} \mathcal{C}_M^{(\n)} ,
\end{align}
which represents all the indices of the signal $y$ that are located at most $L$ indices relative to the beginning of any one of the signals $\ppp{x_\n}_{\n=0}^{\N-1}$. Note that $\abs{\mathcal{C}_M^{(i)}}/L = 2$.
Therefore, we have,
\begin{align}
    \sum_{\jj=0}^{LM-1} & F\pp{\jj} = \sum_{\jj \in \mathcal{C}_M} F\pp{\jj} + \sum_{\jj \in \mathcal{C}_M^c} F\pp{\jj} 
    \\ & = \sum_{k=-L}^{L-1} \sum_{\n=0}^{N-1} F\pp{r_\n + k} + \sum_{\jj \in \mathcal{C}_M^c} F\pp{\jj}. \label{eqn:splittingToDiffernetSums}
\end{align}

Following Lemma \ref{lemma:C0}, $\ppp{F\pp{r_\n + k}}_\n$ form an i.i.d sequence, for every $-L \leq k < L$. Thus, we may apply the strong law of large numbers (SLLN) for $\N, M \to \infty$,
\begin{align}
    \frac{1}{LM} \sum_{\n=0}^{\N-1} F\pp{r_\n + k} \xrightarrow{\text{a.s.}}  \frac{\gamma}{L} \mathbb{E} \pp{F\pp{r_0 + k}}, \label{eqn:A17}
\end{align}
where $\gamma = \N/M$.

As \eqref{eqn:A17} holds for every $-L \leq k < L-1$, we use the general fact that for sequences $\ppp{X_n}_{n \in \mathbb{N}}$, $\ppp{Y_n}_{n \in \mathbb{N}}$, if $X_n \xrightarrow{\text{a.s.}} X$, and $Y_n \xrightarrow{\text{a.s.}} Y$, then $X_n + Y_n \xrightarrow{\text{a.s.}} X + Y$, to show that,
\begin{align}
    \sum_{k=-L}^{L-1} & \frac{1}{LM} \sum_{\n=0}^{N-1} F\pp{r_\n + k} \xrightarrow{\text{a.s.}}  
    \\ & \frac{\gamma}{L}  \sum_{k=-L}^{L-1}\mathbb{E} \pp{F\pp{r_0 + k}}. \label{eqn:A19}
\end{align}

Finally, the last term in \eqref{eqn:splittingToDiffernetSums} can be expressed by,
\begin{align}
    \sum_{\jj \in \mathcal{C}^c} F\pp{\jj} = \sum_{\jj \in \mathcal{C}^c} \epsilon \pp{\jj} \epsilon  \pp{\jj+\ell_1} ...\epsilon \pp{\jj+\ell_{d-1}}.
\end{align}
By applying a similar argument as before, it is clear that
\begin{align}
    \frac{1}{LM} & \sum_{\jj \in \mathcal{C}^c} \epsilon \pp{\jj} \epsilon  \pp{\jj+\ell_1} ...\epsilon \pp{\jj+\ell_{d-1}} \xrightarrow{\text{a.s.}}
    \\ & \p{1 - 2\gamma} \mathbb{E} \pp{\epsilon \pp{0} \epsilon  \pp{\ell_1} ...\epsilon \pp{\ell_{d-1}}}, \label{eqn:A22}
\end{align}
as $M \to \infty$. Substituting \eqref{eqn:A19} and \eqref{eqn:A22} into \eqref{eqn:splittingToDiffernetSums} completes the proof.
\end{proof}

\begin{proof}[Proof of Lemma~\ref{lemma:C2}]
By the definition of $F\pp{\jj}$ \eqref{eqn: app_C0}, and the indices $\ppp{r_\n}_{\n=0}^{\N-1}$, which mark the starting position of the $\n$-th signal, for every \( 0 \leq k < L \), the term \( F\pp{r_\n + k} \) depends solely on the vector \( A_\n \), as defined in \eqref{eqn:seqAi}. In addition, for every \( -L \leq k < 0 \), the term \( F\pp{r_\n + k} \) depends solely on the vector \( B_\n \), as defined in \eqref{eqn:seqBi}. 

Recall the definitions of the random vectors $A$ \eqref{eqn:ArV}, and $B$ \eqref{eqn:BrV}, and their relation to the ensemble random vector $Y$ \eqref{eqn:ensembeleRV}, given in \eqref{eqn:app_A23} and \eqref{eqn:app_A24}.
Then, from \eqref{eqn:app_A13} and \eqref{eqn:app_A15}, the following holds, 
\begin{enumerate}
    \item For every $0 \leq k \leq L-1$,
        \begin{align}
            \nonumber \mathbb{E} & \pp{F\pp{r_0 + k}} =
            \\ & \mathbb{E} \pp{A\pp{k} A\pp{k+\ell_1} ...A\pp{k+\ell_{d-1}}}. \label{eqn:app_A23}
        \end{align}
    \item For every $-L \leq k < 0$,
        \begin{align}
            \nonumber \mathbb{E} & \pp{F\pp{r_0 + k}} =
            \\ & \mathbb{E} \pp{B\pp{k} B\pp{k+\ell_1} ...B\pp{k+\ell_{d-1}}}. \label{eqn:app_A24}
        \end{align}
\end{enumerate}
Combining \eqref{eqn:app_A23} and \eqref{eqn:app_A24} with \eqref{eqn:app_A5} and \eqref{eqn:app_A6}, and summing over $-L \leq k < L$ completes the proof.
\end{proof}

\subsection{Proof of Proposition \ref{thm:prop1}}
To establish this proposition, our objective is to show that the autocorrelations uniquely determine the orbit, with a sample complexity of $\omega(\sigma^{2\bar{d}})$.
There are two key points in this proposition:
\begin{enumerate}
    \item In proposition \ref{thm:prop0}, we addressed the convergence of each entry in the empirical autocorrelation individually. Now, we must examine the convergence of the entire empirical autocorrelation tensor (in the Frobenius norm) to the ensemble mean.

    \item Even if the empirical autocorrelation converges to the ensemble mean with a sample complexity of $\omega \p{\sigma^{2d}}$, this does not guarantee that the orbit estimator converges at the same sample complexity to the orbit of the signal $x$. To achieve this, it is essential that some additional mild conditions would hold, as proved in Lemma \ref{lemma:D3}.
\end{enumerate}

Formally, let \( z \in \mathbb{R}^{LM} \) represent observations sampled from the well-separated MTD model~\eqref{eqn:MTDmodelGeneral}, and denote by $a_z^{({d})}$ the empirical ${d}$-order  autocorrelation as defined in \eqref{eqn:autoCorrealtionMoments}. By Proposition \ref{thm:prop0}, the empirical autocorrelation $a_z^{({d})}$ converge almost surely to
\begin{align}
    \nonumber \lim_{\N,M \to \infty} & {a}^{(d)}_{z} [\ell_1, \ell_2, ..., \ell_{d-1}]
    \\ \nonumber \equalityAS & \ 2 \gamma \bar{a}^{(d)}_{Y,\rho} [\ell_1, \ell_2, ..., \ell_{d-1}]
    \\ & + \p{1-2 \gamma} \chi [\ell_1, \ell_2, ..., \ell_{d-1}], \label{eqn:A23}
\end{align}
for every $0 \leq \ell_1, \ell_2, \ell_{d-1} \leq L-1$ where $\bar{a}_{Y, \rho}^{({d})}$ is the autocorrelation ensemble mean, and $\chi [\ell_1, \ell_2, ..., \ell_{d-1}]$ is a bias term independent of $x$. 
Let us denote by $\bar{\mu}_{G \cdot x}^{(d)} [\ell_1, \ell_2, ..., \ell_{d-1}]$ the right-hand-side of \eqref{eqn:A23}:
\begin{align}
    \nonumber \bar{\mu}_{G \cdot x}^{(d)} [\ell_1, \ell_2, ..., \ell_{d-1}] & = 
     2 \gamma \bar{a}^{(d)}_{Y,\rho} [\ell_1, \ell_2, ..., \ell_{d-1}]
    \\ & + \p{1-2 \gamma} \chi [\ell_1, \ell_2, ..., \ell_{d-1}] \label{eqn:app_A29}
\end{align}
Then, we have the following corollary, which is proved at the end of the section.

\begin{corollary} \label{lemma:D1} 
Let ${a}^{(d)}_{z} [\ell_1, \ell_2, ..., \ell_{d-1}]$ be a single entry $\pp{\ell_1, \ell_2, \dots \ell_{d-1}}$ of the empirical autocorrelation of $d$ order (Definition \ref{def:autoCorrelationMomentsEmpirical}), and $\bar{\mu}_{G \cdot x}^{(d)} [\ell_1, \ell_2, ..., \ell_{d-1}]$ defined by \eqref{eqn:app_A29}. Then,
\begin{align}
    \mathbb{P} \ppp{\abs{{a}^{(d)}_{z} - \bar{\mu}_{G \cdot x}^{(d)}} \geq \delta} \leq O \p{\frac{\sigma^{2d}}{N\delta^2}}. \label{eqn:A29}
\end{align}
\end{corollary} 

Before the next corollary, we define the Frobenius norm between the empirical autocorrelation ${a}^{(d)}_{z}$, and the tensor $\bar{\mu}_{G \cdot x}^{(d)}$, defined by the entries of \eqref{eqn:app_A29}:
\begin{align}
    \nonumber & \norm{{a}^{(d)}_{z} - \bar{\mu}_{G \cdot x}^{(d)}}_F^2 \triangleq \\ & \sum_{\ell_1, \dots \ell_{d-1} = 0}^{L-1} \p{{a}^{(d)}_{z} [\ell_1, ..., \ell_{d-1}] - \bar{\mu}_{G \cdot x}^{(d)} [\ell_1, ..., \ell_{d-1}]}^2. \label{eqn:A30}
\end{align}
Then, we have the following corollary.
\begin{corollary} \label{lemma:D2} 
Let ${a}^{(d)}_{z}$ be the empirical autocorrelation of $d$ order (Definition \ref{def:autoCorrelationMomentsEmpirical}), and $\bar{\mu}_{G \cdot x}^{(d)}$ a tensor composed of entries defined by \eqref{eqn:app_A29}. Then, 
\begin{align}
    \mathbb{P} \ppp{\norm{{a}^{(d)}_{z} - \bar{\mu}_{G \cdot x}^{(d)}}_F \geq \delta} \leq O \p{\frac{\sigma^{2d} L^{d-1}}{N\delta^2}}. \label{eqn:A32}
\end{align}
\end{corollary} 

We are now ready to state the main result of this proposition:
\begin{lem} \label{lemma:D3} 
Assume the ensemble means of the autocorrelations (Definition \ref{def:autoCorrelationMomentsEnsembele}) up to order \(\bar{d}\) almost everywhere determine the orbit of the group action \(G\). Assume that the parameter space $x \in \Theta$ of the model is compact. Recall the definition of $\s{MSE}^*_{\s{MTD}}(\sigma^2, \N)$, as defined in \eqref{eqn:sampleComplexityMTD}. Then, for  $\N = \omega\p{\sigma^{2\bar{d}}}$, 
\begin{align}
    \lim_{\N \to \infty} \s{MSE}^*_{\s{MTD}}(\sigma^2, \N) = 0. \label{eqn:A33}
\end{align}
\end{lem} 
Lemma \ref{lemma:D3} completes the proof of the proposition. It remains to prove the corollaries. 

\begin{proof}[Proof of Corollary~\ref{lemma:D1}]
By Chebyshev's inequality, for a \textit{single} entry of the empirical autocorrelation $0 \leq \ell_1, \ell_2, ... \ell_{d-1} \leq L-1$, we have,
\begin{align}
    \mathbb{P} \ppp{\abs{{a}^{(d)}_{z} - \bar{\mu}_{G \cdot x}^{(d)}} \geq \delta} \leq \frac{\mathbb{E} \pp{\p{{a}^{(d)}_{z} - \bar{\mu}_{G \cdot x}^{(d)}}^2}}{\delta^2}. \label{eqn:A24}
\end{align}

As established in Lemma \ref{lemma:C1}, the variance of the autocorrelation is bounded above by \( O\p{\sigma^{2d}} \), since the number of Gaussian noise terms multiplied is at most \( d \). Consequently, we have:  
\begin{align}
    \mathbb{E} \pp{\abs{{a}^{(d)}_{z} - \bar{\mu}_{G \cdot x}^{(d)}}^2} = O \p{\frac{\sigma^{2d}}{N}}. \label{eqn:A25}
\end{align}
Substituting \eqref{eqn:A25} into \eqref{eqn:A24},
\begin{align}
    \mathbb{P} \ppp{\abs{{a}^{(d)}_{z} - \bar{\mu}_{G \cdot x}^{(d)}} \geq \delta} \leq O \p{\frac{\sigma^{2d}}{N\delta^2}}. \label{eqn:A26}
\end{align}
Thus, the MSE between a single entry of the empirical autocorrelation of order \( d \), ${a}^{(d)}_{z} [\ell_1, \ell_2, ..., \ell_{d-1}]$,  and the corresponding right-hand side of \eqref{eqn:A23}, $\bar{\mu}_{G \cdot x}^{(d)} [\ell_1, \ell_2, ..., \ell_{d-1}]$, is \( O\p{\sigma^{2d}/N} \). 
\end{proof}

\begin{proof}[Proof of Corollary~\ref{lemma:D2}]
To establish the sample complexity for the entire empirical autocorrelation tensor entries in Frobenius norm (comprising $L^{d-1}$ entries of the empirical autocorrelation of order $d$), we employ the Markov's inequality, \eqref{eqn:A30}, and \eqref{eqn:A25} to establish:
\begin{align}
        \nonumber \mathbb{P} & \ppp{\norm{{a}^{(d)}_{z} - \bar{\mu}_{G \cdot x}^{(d)}}_F^2 \geq \delta^2} \leq 
         \\ \nonumber & \leq {\delta^{-2}} \ \mathbb{E} \pp{\norm{{a}^{(d)}_{z} - \bar{\mu}_{G \cdot x}^{(d)}}_F^2}
        \\ & \leq  O \p{\frac{\sigma^{2d}L^{d-1}}{N\delta^2 }}. \label{eqn:A26_2}
\end{align}
\end{proof}

\begin{proof}[Proof of Lemma~\ref{lemma:D3}]
To prove the lemma, we reformulate the problem as a likelihood estimation problem of the orbit $G \cdot x$ from all the empirical autocorrelations up to order $\bar{d}$, and apply well-established results from the literature on likelihood estimation \cite{newey1994chapter}.  
Under the assumptions of the proposition, the ensemble means of the autocorrelations up to order \(\bar{d}\) determine almost everywhere the orbit of the group action \(G\). Let \( z \in \mathbb{R}^{LM} \) denote an observation sampled from the well-separated MTD model~\eqref{eqn:MTDmodelGeneral}.

Let \( h: \mathbb{R}^{LM} \to U \) denote the mapping between \( z \) to the empirical autocorrelations up to order \( \bar{d} \), defined as:  
\begin{align}
    h\p{z; G \cdot x} = \p{{a}^{(\bar{d})}_{z (G \cdot x)}, {a}^{(\bar{d}-1)}_{z (G \cdot x)}, \dots, {a}^{(1)}_{z (G \cdot x)}}, \label{eqn:A_37}
\end{align}  
where \( G \cdot x \) explicitly represents the dependence of \( z \) on the orbit \( G \cdot x \).  

Let \( \hat{Q}_M\p{z; G \cdot x} \) denote the distribution of \( h\p{z; G \cdot x} \), where \( M \) indicates its dependence on the observation length \( z \). Then, we state the estimation problem of the orbit of $x$ as follows: 
\begin{align}
    \widehat{G \cdot x} = \argmax_{G \cdot x, x \in \Theta} \hat{Q}_M\p{z; G \cdot x},
\end{align}
that is, the orbit of $x$ which maximizes the likelihood function $\hat{Q}_M\p{z; G \cdot x}$.

We note the following properties of $\hat{Q}_M\p{z; G \cdot x}$:
\begin{enumerate}
    \item \textit{Identification} - By the assumption that the orbit is uniquely defined by all the autocorrelations up to order $\bar{d}$, it means that $\hat{Q}_M\p{z; G \cdot x}$ is identifiable, that is, for every $x \neq x_0$, $\hat{Q}_M\p{z; G \cdot x} \neq \hat{Q}_M\p{z; G \cdot x_0}$.
    \item \textit{Convergence of $\hat{Q}_M$} - By corollary \ref{lemma:D2}, for $N = \omega \p{\sigma^{2\bar{d}}}$, as $N \to \infty$,
        \begin{align}
           \hat{Q}_M\p{z; G \cdot x} \xrightarrow{\text{a.s.}} Q_0\p{G \cdot x}, \label{eqn:app_A40}
        \end{align}
    where $ Q_0\p{G \cdot x} $ indicates the expected value of the right-hand-side of \eqref{eqn:A_37}, given by $\p{\bar{\mu}_{G \cdot x}^{(\bar{d})}, \dots \bar{\mu}_{G \cdot x}^{({1})}}$, where $\bar{\mu}_{G \cdot x}^{({d})}$ is defined by \eqref{eqn:app_A29}.
    \item \textit{Continuity of $\hat{Q}_M \p{z; G\cdot x}$ in $x$} - Since $\hat{Q}_M$ is polynomial in $x$, it is also continuous in $x$.
    \item \textit{Compactness} - By assumption, the parameter space $\Theta$ is compact.
\end{enumerate}

Then, by applying ~\cite[Theorem 2.1]{newey1994chapter}, where all the condition of the theorem are satisfied, we obtain the result:
\begin{align}
    \widehat{G \cdot x} \xrightarrow{\mathcal{P}} G \cdot x. \label{eqn:A_41}
\end{align}
Since \eqref{eqn:app_A40} holds for $N = \omega \p{\sigma^{2 \bar{d}}}$, \eqref{eqn:A_41} also holds under this condition, which completes the proof of the lemma. 
\end{proof}

\subsection{Proof of Proposition \ref{thm:prop2}}
The key to proving the proposition is to construct a measurable mapping \( h \) between the MRA model and the corresponding MTD model, such that the image of this mapping captures the statistics of the MTD model. If such a mapping exists, it implies that the MRA model contains at least as much statistical information as the MTD model. As a result, the MSE of the MTD model is lower-bounded by the MSE of the MRA model.

Formally, denote $N$ i.i.d.\ random vectors $\ppp{Z_\n}_{\n=0}^{\N-1}$, where $ Z_\n \in \mathbb{R}^L$ are instances of the MRA model defined in \eqref{eqn:MRA_model} 
Then, we define a map between $\N$ MRA instances to the MTD observation by $Y = h\p{\ppp{Z_\n}_{\n=0}^{\N-1}} \in \mathbb{R}^{LM}$, as follows,
\begin{enumerate}
    \item Initialize $Y = 0$.
    \item Place the MRA instances $\ppp{Z_\n}_{\n=0}^{\N-1}$ within the vector $Y$ in any valid non-overlapping configuration (Definition \ref{def:wellSeperatedModel}).
    \item In all the remaining indices of $Y$, which do not contain the instances $\ppp{Z_\n}_{\n=0}^{\N-1}$, add additive i.i.d random noise with variance $\sigma^2$.
\end{enumerate}
Clearly, the map defined above is a measurable map and defines the non-overlapping MTD model in \eqref{eqn:MTDmodelGeneral}, according to \textit{any} unknown locations $\ppp{s_\n}_{\n=0}^{\N-1}$. Lemma \ref{lemma:1}, proved below, would show that
\begin{align}
    \s{MSE}^*_{\s{MRA}}(\sigma^2, \N) \leq  \s{MSE}^*_{\s{MTD}}(\sigma^2, \N), \label{eqn:app_B0}
\end{align}
where $\s{MSE}^*_{\s{MRA}}(\sigma^2, \N)$ is the minimal mean square error of any estimator of the MRA model, and $\s{MSE}^*_{\s{MTD}}(\sigma^2, \N)$ is defined in \eqref{eqn:infMtdMSE}.

\begin{lem} \label{lemma:1} 
 Let $\ppp{U_i}_{i=0}^{n-1} \sim f_U \p{\cdot ; \theta}$, and $\ppp{V_i}_{i=0}^{n-1} \sim f_V \p{\cdot ; \theta}$, two sets of observations following the distribution laws $f_U$, and $f_V$ respectively, which are governed by the same parameter $\theta$ from the a parameter space $\theta \in \Theta$. Let $\hat{\theta}_U$, and $\hat{\theta}_V$ be the estimators of $\theta$ by the observations $\ppp{U_i}_{i=0}^{n-1}$ and $\ppp{V_i}_{i=0}^{n-1}$, respectively. Then, if $V = h\p{U}$ for a measurable function $h$, we have, 
\begin{comment}
\begin{align}
    \nonumber \inf_{\hat{\theta}} \mathbb{E} & \pp{ \s{MSE} \p{\hat{\theta}\p{\ppp{U_i}_{i=0}^{n-1}}, \theta}} \\ & \leq \inf_{\hat{\theta}} \mathbb{E} \pp{ \s{MSE} \p{\hat{\theta}\p{\ppp{V_i}_{i=0}^{n-1}}, \theta}}
\end{align}
\end{comment}
\begin{align}
    \nonumber \inf_{\hat{\theta}_U} \mathbb{E}_U & \pp{ \norm{\hat{\theta}_U \p{\ppp{U_i}} - \theta}^2} 
    \\ & \leq \inf_{\hat{\theta}_V} \mathbb{E}_V \pp{ \norm{\hat{\theta}_V \p{\ppp{V_i}} - \theta}^2}. \label{eqn:app_B1}
\end{align}
\end{lem}
\begin{proof}[Proof of Lemma~\ref{lemma:1}]
The proof is similar to the proof of \cite[Theorem 2.93]{ schervish2012theory}.
%\tamir{Is there an off-the-shelf result that you can state?}
First, we show that for every estimator $\hat{\theta}_V$, there is an estimator $\hat{\theta}_U$, such that,
\begin{align}
     \mathbb{E}_U & \pp{ \norm{\hat{\theta}_U \p{\ppp{U_i}} - \theta}^2} 
     = \mathbb{E}_V \pp{ \norm{\hat{\theta}_V \p{\ppp{V_i}} - \theta}^2}. \label{eqn:app_B2}
\end{align}
Define $\hat{\theta}_U \p{\ppp{U_i}} = \hat{\theta}_V \p{\ppp{h \p{U_i}}}$. Then, we have,
\begin{align}
    \nonumber \mathbb{E}_U & \pp{ \norm{\hat{\theta}_U \p{\ppp{U_i}} - \theta}^2} = 
    \\ & = \mathbb{E}_U \pp{ \norm{\hat{\theta}_V \p{\ppp{h\p{U_i}}} - \theta}^2} \label{eqn:app_B3}
    \\ & = \mathbb{E}_V \pp{ \norm{\hat{\theta}_V \p{\ppp{V_i}} - \theta}^2}, \label{eqn:app_B4}
\end{align}
where \eqref{eqn:app_B3} follows from the definition of $\hat{\theta}_U \p{\ppp{U_i}} = \hat{\theta}_V \p{\ppp{h \p{U_i}}}$, and \eqref{eqn:app_B4} follows from the assumption that $V = h\p{U}$.
As \eqref{eqn:app_B2} holds for every estimator $\hat{\theta}_V$, \eqref{eqn:app_B1} follows immediately, which completes the proof.
\end{proof}

Thus, \eqref{eqn:app_B0} follows from Lemma \ref{lemma:1}. Based on the definition of sample complexity in \eqref{eqn:sampleComplexityMTD}, it then follows that,
\begin{align}
     \N^*_{\s{MTD}_G} \p{\sigma^2, \epsilon} \geq  \N^*_{\s{MRA}_G} \p{\sigma^2, \epsilon},
\end{align}
which completes the proof of the proportion.

\subsection{Proof of Proposition \ref{thm:thm1DsignalTranslations}}
First, we prove the lower bound, followed by the upper bound.
For the lower bound, define the 1-D MRA model with circular translations as follows:
\begin{equation}
    z_i = g_i \cdot x + \varepsilon_i, \quad g_i \in \mathbb{Z}_L, \label{eqn:MRAmodel1Dcyclic}
\end{equation}
where \( x \in \mathbb{R}^L \) and the values \( \{g_i\} \) are drawn uniformly from \( \mathbb{Z}_L \). The sample complexity of the MRA model defined in \eqref{eqn:MRAmodel1Dcyclic} is \( \omega \p{\sigma^6} \) \cite{abbe2018multireference}.
Then, by Proposition \ref{thm:prop2}, the associated MTD model satisfies:
\begin{align}
     \N^*_{\s{MTD}_{\mathbb{Z}_L}} \geq  \N^*_{\s{MRA}_{\mathbb{Z}_L}} = \omega \p{\sigma^6},
\end{align}
which completes the proof of the lower bound. 

For the upper bound, it has been shown that the first three orders, \( \bar{d} = 3 \), of the autocorrelations uniquely determine the orbit of $x$, assuming $x$ has a non-vanishing DFT \cite[Theorem 3.2]{bendory2023multi}.  
Therefore, following Proposition \ref{thm:prop1}, the upper bound for the sample complexity is \( \omega \p{\sigma^6} \).

\subsection{Proof of Proposition \ref{thm:thm2Drotations}}
Let us consider the MRA model associated with $\s{SO}(2)$ rotations,
\begin{align}
    z_i = g_i \cdot x + \epsilon_i, \quad g_i \in \s{SO}(2) \label{eqn:mra2Drotations}
\end{align}
where $\ppp{g_i}$ are drawn uniformly from $\s{SO}(2)$, and assuming that $x$ is a band-limited image in Fourier-Bessel expansion \cite{ma2019heterogeneous}. It was shown that a lower bound for the sample complexity of the MRA problem \eqref{eqn:mra2Drotations} is $\omega \p{\sigma^6}$, as the minimal moment required to determine uniquely the signal $x$ is the third moment. This follows from the fact that the combined power spectrum and bispectrum determine an image uniquely, up to global rotation ~\cite[Theorem II.1]{ma2019heterogeneous}, under the assumption of non-vanishing Fourier-Bessel coefficients, as elaborated below. This is in contrast with the power spectrum alone, which does not uniquely determine the image~\cite{ma2019heterogeneous}.

Formally, denote the signal $x$ in polar coordinates by $I \p{r, \theta}$. Then, $I$ can be represented in steerable basis functions in polar coordinates as follows,
\begin{align}
    I(r, \theta) = \sum_{k,q} a_{k,q} u_{k,q}(r, \theta), \label{eqn:image_expansion}    
\end{align}
where,
\begin{align}
    u_{k,q}(r, \theta) = f_{k,q}(r)e^{i k \theta}, \label{eqn:steerable_basis}
\end{align}
are the Fourier-Bessel basis functions, and $a_{k,q}$ are the Fourier-Bessel coefficients, given by,
\begin{align}
    a_{k,q} = \int  I\p{r, \theta} u_{k,q}(r, \theta) \ r dr  d\theta
\end{align}

Then, we use the following theorem \cite[Theorem II.1]{ma2019heterogeneous}:
\begin{thm}[Theorem II.1, ~\cite{ma2019heterogeneous}] \label{thm:2DmraThm}
    Consider two images with steerable basis coefficients $a_{k,q}$ and $a'_{k,q}$, respectively, in the range $-k_{\max} \leq k \leq k_{\max}$. Assume that for all $-k_{\max} \leq k \leq k_{\max}$, there exists at least one $q$ such that $a_{k,q} \neq 0$. If for all indices
    \begin{equation}
    a_{k,q_1} \overline{a_{k,q_2}} = a'_{k,q_1} \overline{a'_{k,q_2}}, \label{eq:condition1}
    \end{equation}
    \begin{equation}
    a_{k_1,q_1} a_{k_2,q_2} \overline{a_{k_1+k_2,q_3}} = a'_{k_1,q_1} a'_{k_2,q_2} \overline{a'_{k_1+k_2,q_3}}, \label{eq:condition2}
    \end{equation}
    then there exists $\theta \in [0, 2\pi)$ such that
    \begin{equation}
    a'_{k,q} = a_{k,q}e^{-\iota k\theta} \label{eq:result}
    \end{equation}
    for all $k, q$. That is, the two images only differ by a rotation.
\end{thm}

Thus, applying Proposition \ref{thm:prop1} and Theorem \ref{thm:2DmraThm} proves the lower bound,
\begin{align}
     \N^*_{\s{MTD}_{\s{SO}(2)}} \geq   \N^*_{\s{MRA}_{\s{SO}(2)}} = \omega \p{\sigma^6}.
\end{align}

\subsection{Proof of Proposition \ref{thm:thm1Dsignal}}
It has been proved that the first three orders of autocorrelations (\( \bar{d} = 3 \)) uniquely determine a generic signal \( x \), which has non-vanishing edge entries, i.e., $x\pp{0}, x\pp{L-1} \neq 0, $ ~\cite[Proposition 4.1]{bendory2019multi}. Therefore, following Proposition \ref{thm:prop1}, the upper bound of the sample complexity is \( \omega \p{\sigma^6} \).

%\end{appendices}

\end{document}